\documentclass{article}%
\pdfoutput=1
\usepackage{fullpage}
\usepackage{amsfonts}
\usepackage{amsmath}
\usepackage{amsthm}
\usepackage{amssymb}
\usepackage{graphicx}
\usepackage{algorithm}
\usepackage{algorithmic}%
\providecommand{\U}[1]{\protect\rule{.1in}{.1in}}
\newtheorem{theorem}{Theorem}

\newtheoremstyle{example}{\topsep}{\topsep}
{}
{}
{\bfseries}
{}
{  }
{\thmname{#1}\thmnumber{ #2.}\thmnote{ (#3)}}
\theoremstyle{example}

\begin{document}

\title{Minimal-memory realization of pearl-necklace encoders of general quantum convolutional codes}
\author{Monireh Houshmand and Saied Hosseini-Khayat\\\textit{Department of Electrical Engineering,}\\\textit{Ferdowsi University of Mashhad, Iran}\\\{monirehhoushmand, saied.hosseini\}@gmail.com}

\date{\today}
%\{monirehhoushmand,saied.hosseini,mwilde\}gmail.com

\maketitle
\begin{abstract}
Quantum convolutional codes, like their classical counterparts,
promise to offer higher error correction performance than block codes of equivalent 
encoding complexity,
and are expected to find important applications in reliable quantum communication 
where a continuous stream of
qubits is transmitted.
Grassl and Roetteler devised an algorithm to encode a quantum convolutional code 
with a ``pearl-necklace encoder."
Despite their theoretical significance as a neat way of representing quantum
convolutional codes, they are not well-suited to practical realization.
In fact, there is no straightforward way to implement any given pearl-necklace structure.
This paper closes the gap between theoretical representation and practical
implementation.
In our previous work, we presented an efficient algorithm for finding a minimal-memory
realization of a pearl-necklace encoder for Calderbank-Shor-Steane (CSS) convolutional codes.
This work extends our previous work and presents an
algorithm for turning a pearl-necklace encoder for a general (non-CSS) quantum
convolutional code into a realizable quantum convolutional encoder.
We show that a minimal-memory realization depends on the commutativity
relations between the gate strings in the pearl-necklace encoder.
We find a realization by means of a weighted graph which details the
non-commutative paths through the pearl-necklace.
The weight of the longest path in this graph is equal to the minimal amount of
memory needed to implement the encoder.
The algorithm has a polynomial-time complexity in the number of gate strings in
the pearl-necklace encoder.
\end{abstract}
\begin{figure}
[ptb]
\begin{center}
\includegraphics[
natheight=3.499900in,
natwidth=11.000400in,
width=6.0502in
]%
{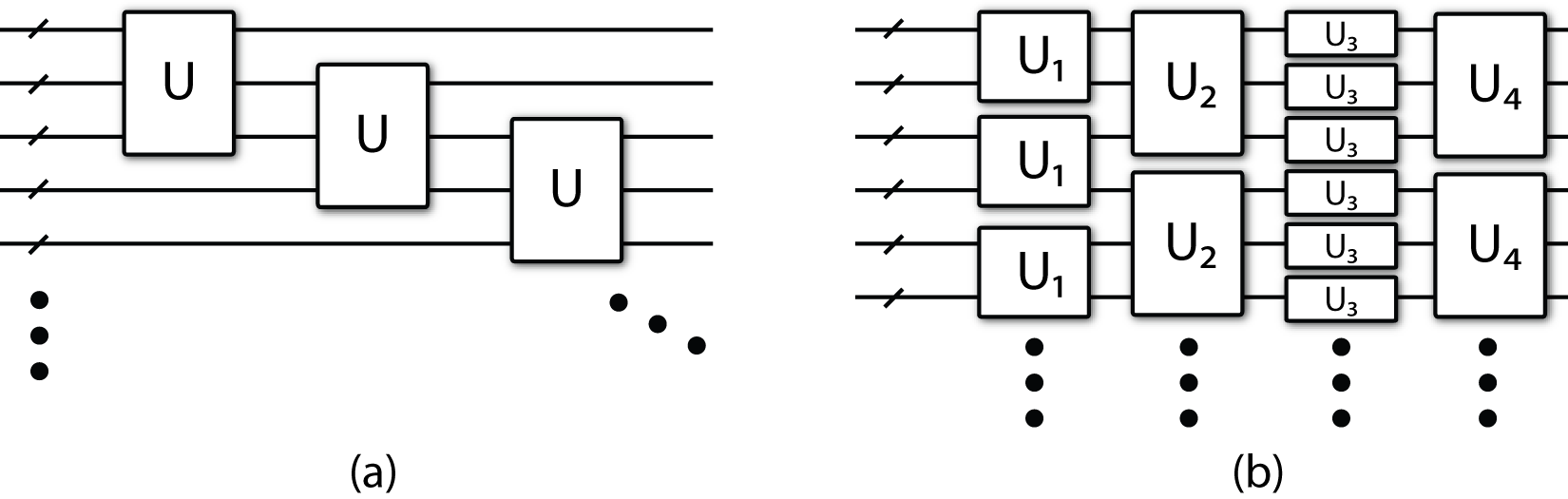}%
\caption{Two different representations of the encoder for a quantum
convolutional code. (a) Representation of the encoder as a convolutional
encoder. (b) Representation of the encoder as a pearl-necklace encoder~\cite{ourpaper}.}%
\label{qcc-and-pearl-necklace}%
\end{center}
\end{figure}
Quantum error correction codes are used to protect quantum information from decoherence and operational errors~\cite{qecbook,PhysRevLett.81.2594,thesis97gottesman,book2000mikeandike,PhysRevLett.77.793,PhysRevA.54.1098,PhysRevLett.78.405,ieee1998calderbank,PhysRevLett.79.3306,mpl1997zanardi}. Depending on their approach to error control,
error correcting codes can be divided into two general classes: block codes and convolutional codes. In the case of a block code, the original state is first divided into a finite number of blocks of fixed length.
Each block is then encoded separately and the encoding is independent of the other blocks.
On the other hand, a quantum convolutional code~\cite{PhysRevLett.91.177902,arxiv2004olliv,isit2006grassl,ieee2006grassl,ieee2007grassl,isit2005forney,ieee2007forney,cwit2007aly,arx2007aly,arx2007wildeCED,arx2007wildeEAQCC,arx2008wildeUQCC,arx2008wildeGEAQCC,pra2009wilde}
encodes an incoming stream of quantum information into an
outgoing stream. Fast decoding algorithms exist for quantum convolutional codes~\cite{arx2007poulin} and in general, they are preferable in terms of their performance-complexity
tradeoff~\cite{ieee2007forney}.
%and also  them, hence they suit fault-tolerant implementation.

The
encoder for a quantum convolutional code has a representation as a \emph{convolutional encoder} or as a \emph{pearl-necklace} encoder. The convolutional encoder ~\cite{PhysRevLett.91.177902,arxiv2004olliv},~\cite{arx2007poulin} consists of a single unitary
repeatedly applied to a stream of quantum data (see Figure~\ref{qcc-and-pearl-necklace}(a)).
On the other hand, the pearl-necklace encoder (see Figure~\ref{qcc-and-pearl-necklace}(b)) consists of several strings
of the same unitary applied to the quantum data stream. Grassl
and R\"{o}tteler~\cite{isit2006grassl} proposed an algorithm for encoding any
quantum convolutional code with a pearl-necklace encoders. The algorithm consists of a sequence
of elementary encoding operations. Each of these
elementary encoding operations corresponds to a gate string in the pearl-necklace encoder.

The amount of required memory plays a key role for implementation of any
encoder, since this amount will result in overhead in the implementation of communication protocols. Hence any reduction in the required amount of memory will help in practical implementation of quantum computer.

It is trivial to determine the amount of memory required for implementation of a convolutional encoder: it is equal to the number of qubits that are
fed back into the next iteration of the unitary that acts on the stream. For example, the convolutional encoders in the Figures~\ref{qcc-and-pearl-necklace}(a),~\ref{pearl2convosimple}(c) and~\ref{pearl-convo-ex}(b) require two, one and four frames
of memory qubits, respectively.

In contrast, the practical realization of a pearl-necklace encoder is not explicitly clear. To make it realizable, one should rearrange the gate strings in the pearl-necklace encoder so that it becomes a convolutional encoder.
%There are many convolutional encoders that encode the same code as a given pearl-necklace encoder, but the one with the minimal amount of memory requirement is of much interest.
 In~\cite{ourpaper} we proposed an algorithm for finding the minimal-memory realization of a pearl necklace encoder for the CSS class of convolutional codes. This kind of encoder consists of CNOT gate strings only \cite{PhysRevA.79.062325}.

In this paper we extend our work to find the minimal-memory realization of a pearl-necklace encoder for a general (non-CSS) convolutional code. A general case includes all gate strings that are in the
shift-invariant Clifford group~\cite{isit2006grassl}:
 Hadamard gates, phase gates,  controlled-phase
gate string, finite-depth and infinite-depth~\cite{arx2008wildeUQCC,pra2009wilde} CNOT operations. We show that there are many realization for a given pearl-necklace encoder which are obtained considering non-commutativity relations of gate strings in the pearl-necklace encoder.
%We define three types of non-commutativity each having its own influence on the required memory.
 Then for finding the minimal-memory realization a specific graph, called non-commutativity graph is introduced. Each vertex in the non-commutativity graph, corresponds to a gate string in the pearl-necklace encoder. The graph features a directed edge from one vertex
to another if the two corresponding gate strings do not commute. The weight of a directed edge
depends on the degrees of the two corresponding gate strings and their type of non-commutativity. The weight of the longest path in the graph is equal to the minimal memory requirement for the pearl-necklace encoder. The complexity for constructing this graph is quadratic in
the number of gate strings in the encoder.

The paper is organized as follows. 
In Section~\ref{sec:def-not}, we introduce some definitions and notation 
that are used in the rest of paper. 
In Section~\ref{memory-general}, we define three different types 
of non-commutativity and then propose an algorithm to find the  minimal memory 
requirements in a general case. In Section \ref{conclu}, we will summarize 
the contribution of this paper.

\section{Definitions and notation}

\label{sec:def-not}We first provide some definitions and notation which are useful for our analysis later on.
The gate strings in the pearl-necklace encoder and the gates in the convolutional encoder are numbered from left to right. We denote the $i^{\text{th}}$ gate string in the pearl-necklace encoder,
$\overline{U}_i,$ and the $i^{\text{th}}$ gate in the convolutional encoder, $U_i.$

Let $\overline{U}$, without any index
specified, denote a particular infinitely repeated sequence of $U$ gates, where the sequence contains the
same $U$ gate for every frame of qubits.

Let $U$ be either CNOT or CPHASE gate. The notation $\overline{U}\left(  a,bD^{l}\right)$ refers to a string of gates in a pearl-necklace encoder and
denotes an infinitely repeated sequence of $U$ gates from qubit $a$ to qubit
$b$ in every frame where qubit $b$ is in a frame delayed by $l$.
\footnote{Instead of the previously used notation $\overline{U}(  a,b)(D^{l}),$ 
we prefer to use  $\overline{U}\left(  a,bD^{l}\right)$  
as it seems to better represent the concept.}

Let $U$ be either phase or Hadamard gate. The notation $\overline{U}\left(b\right)  $ refers to a string of gates in a pearl-necklace encoder and
denotes an infinitely repeated sequence of $U$ gates which act on  qubit $b$ in every frame.
By convention we call this qubit, the target of $\overline{U}\left(b\right)  $ during this paper.

If $\overline{U_i}$ is $\overline{\text{CNOT}}$ or $\overline{\text{CPHASE}}$ the notation $a_i,b_i,$ and $l_i$ are used to denote its source index, target index and degree, respectively. If $\overline{U_i}$ is $\overline{H}$ or $\overline{P}$ the notation $b_i$ is used to denote its target index.

For example,
the  strings of gates in Figure \ref{pearl2convosimple}(a)
%and \ref{pearl-convo-ex}(a)
correspond
to: \begin{align}
 \overline{H}\left(  3\right)     \overline{\text{CPHASE}}%
\left(  1,2D\right)  \overline{\text{CNOT}}\left(  1,3\right),
\end{align}
 $b_1=1$, $a_2=1,$ $b_2=2,$ $l_2=1,$ $a_3=1,$ $b_3=3,$ and $l_3=0.$

Suppose the number of gate strings in the pearl-necklace encoder is $N.$ The members of the sets $I_{\text{CNOT}}^{+}$, $I_{\text{CNOT}}^{-},$ $I_{\text{CPHASE}}^{+},$ and $I_{\text{CPHASE}}^{-}$ are the indices of gate strings in the encoder which are $\overline{\text{CNOT}}$ with non-negative degree, $\overline{\text{CNOT}}$ with negative degree, $\overline{\text{CPHASE}}$ with non-negative degree and $\overline{\text{CPHASE}}$ with negative degree respectively:
\[I_{\text{CNOT}}^{+}=\{i|\,\overline{U}_i\,\, \text{is}\,\,\overline{\text{CNOT}},l_i\geq 0, i\in \{1,2,\cdots,N\}\},\]
\[I_{\text{CNOT}}^{-}=\{i|\,\overline{U}_i\,\, \text{is}\,\,\overline{\text{CNOT}},l_i< 0, i\in \{1,2,\cdots,N\}\},\]
\[I_{\text{CPHASE}}^{+}=\{i|\,\overline{U}_i \,\,\text{is}\,\,\overline{\text{CPHASE}},l_i\geq 0, i\in \{1,2,\cdots,N\}\},\]
\[I_{\text{CPHASE}}^{-}=\{i|\,\overline{U}_i \,\,\text{is}\,\,\overline{\text{CPHASE}},l_i< 0, i\in \{1,2,\cdots,N\}\}.\]
The members of the sets $I_{\text{H}}$ and $I_{\text{P}}$ are the indices of gate strings of the encoder which are $\overline{H}$ and $\overline{P}$ respectively:
\[I_{\text{H}}=\{i|\,\overline{U}_i \,\,\text{is}\,\,\overline{H}, i\in \{1,2,\cdots,N\}\},\]
\[I_{\text{P}}=\{i|\,\overline{U}_i \,\,\text{is}\,\,\overline{P}, i\in \{1,2,\cdots,N\}\}.\]

Our convention for numbering the
frames upon which the unitary of a convolutional encoder acts is from
\textquotedblleft bottom\textquotedblright\ to \textquotedblleft
top.\textquotedblright\ Figure \ref{graph-convo1}(b) illustrates this convention for a convolutional encoder. If ${U}_i$ is CNOT or CPHASE gate, then let $\sigma_{i}$ and $\tau_{i}$ denote the frame index of the
respective source and target qubits of the ${U}_i$ gate in a
convolutional encoder. If $U_i$ is Hadamard or Phase gate, let $\tau_{i}$ denote the frame index of the target qubit of the ${U}_i$ gate in a convolutional encoder. For example, consider the convolutional encoder in Figure \ref{graph-convo1}(b). The convolutional encoder in this figure consists of six gates; $\tau_{1}=0,$ $\tau_{2}=0,$ $\sigma_{3}=0,$ $\tau_{3}=1,$ $ \sigma_{4}=2,$ $\tau_{4}=0,$ $ \sigma_{5}=3,$ $\tau_{5}=2,$ $\sigma_{6}=4, $ and $\tau_{6}=3.$

While referring to a convolutional encoder, the following notation are defined as follows:\\
The notation CNOT$(a,b)\left(  \sigma,\tau \right)$ denotes a CNOT gate from qubit $a$ in
frame $\sigma$ to qubit $b$ in frame $\tau$.\\
The notation CPHASE$(a,b)\left(  \sigma,\tau \right) $ denotes a CPHASE gate from qubit $a$ in
frame $\sigma$ to qubit $b$ in frame $\tau$.\\
The notation $H(b)(\tau)$ denotes a Hadamard gate which acts on qubit $b$ in frame $\tau.$\\
The notation $P(b)(\tau)$ denotes a Phase gate which acts on qubit $b$ in frame $\tau.$

For example the gates in Figure \ref{graph-convo1}(b) correspond to:\begin{align}
 &H\left(  1\right)(0)   P\left(  1\right)(0)  \text{CPHASE}%
\left(  1,2\right)\left(0,1 \right) \text{CPHASE}%
\left(  2,3\right) \left( 2,0\right)
  \text{  CNOT}\left(
3,2\right)  \left( 3,2 \right)  \text{CNOT}\left(  2,3\right)  \left( 4,3 \right)
\nonumber.
\end{align}
\section{Memory requirements for an arbitrary pearl-necklace encoder}
\label{memory-general}As discussed before, for finding the practical realization of a pearl-necklace encoder it is required to rearrange the gates as a convolutional encoder.

To do this rearrangement, we must first find a
set of gates consisting of a single gate for each gate string
in the pearl-necklace encoder
%(the first $U$ in the Figure~\ref{qcc-and-pearl-necklace}(a)),
 such that all the gates that remain after the set
commute with it. Then we can shift all these gates to the right and infinitely repeat this operation on
the remaining gates to obtain a convolutional encoder. When all gates in the pearl-necklace encoder commute with each other, there is no constraint on frame indices of target (source) qubits of gates in the convolutional encoder \cite{ourpaper}. (Figure~\ref{pearl2convosimple} shows an  example
of the rearrangement of commuting gate strings into a convolutional encoder.) On the other hand, when the gate strings do not commute, the constraint of commutativity of the remaining gates with the chosen set results in constraints on frame indices of target (source) qubits of gates in the convolutional encoder.

In the following sections, after defining different types of non-commutativity and their imposed constraints, the algorithm for finding the minimal-memory convolutional encoder for an arbitrary pearl-necklace encoder is presented.

 \begin{figure}
[ptb]
\begin{center}
\includegraphics[
natheight=9.879900in,
natwidth=19.139999in,
height=2.35in,
width=6.24in
]
{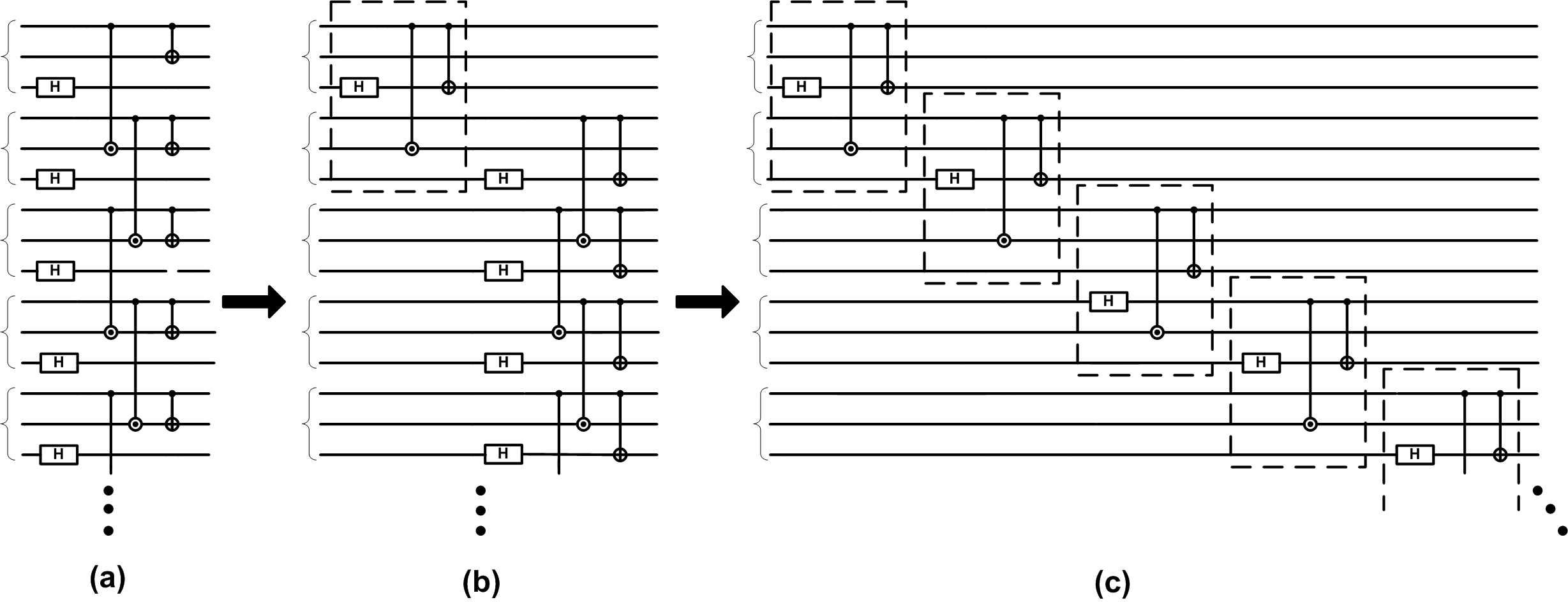}
\end{center}
\caption{Simple (since all gate strings commute with each other) example of the rearrangement of a pearl-necklace encoder for a non-CSS code
into a convolutional encoder. (a) The pearl-necklace encoder consists of
the gate strings
$\overline{H}\left(  1\right)     \overline{\text{CPHASE}}
\left(  1,2\right)\left(  D\right)  \overline{\text{CNOT}}\left(  1,3\right)  $.
(b) The
rearrangement of gates after the first three by shifting them
to the right. (c) The repeated application of the procedure in (b) realizes a
convolutional encoder from a pearl-necklace encoder.}
\label{pearl2convosimple}
\end{figure}

\subsection{Different types of non-commutativity and their imposed constraints}
%
% In the following, all  non-commutativities that may exist between any two gate strings of shift-invariant Clliford and the constraints that they impose on frame indices of gates in the convolutional encoder are stated.
 There may arise three types of non-commutativity for any two gate strings of shift-invariant Clifford: source-target non-commutativity, target-source non-commutativity and target-target non-commutativity. Each imposes a different constraint on frame indices of gates in the convolutional encoder. These types of non-commutativity and their constraints are explained in the following sections.
\subsubsection{Source-target non-commutativity}
The gate strings in~(\ref{cn-cn-s-t1}-\ref{cp-h-s-t1}) below do not commute with each other. In all of them, the index of each source qubit in the first gate string is the same as the index of each target qubit in the second gate string, therefore we call this type of non-commutativity  \emph{source-target non-commutativity}.
\begin{equation}
\overline{\text{CNOT}}(a,bD^{l}) \overline{\text{CNOT}}(a^{\prime},b^{\prime}D^{l^{\prime}}),\,\,\text{where}\,\, a=b^{\prime},\label{cn-cn-s-t1}
\end{equation}
\begin{equation}
\overline{\text{CPHASE}}(a,bD^{l}) \overline{\text{CNOT}}(a^{\prime},b^{\prime}D^{l^{\prime}}),\,\,\text{where}\,\,a=b^{\prime}, \label{cp-cn-s-t1}
\end{equation}
\begin{equation}
\overline{\text{CNOT}}(a,bD^{l})\overline{H}(b^{\prime}),\,\,\text{where}\,\, a=b^{\prime},\label{cn-h-s-t1}
\end{equation}
\begin{equation}
\overline{\text{CPHASE}}(a,bD^{l})\overline{H}(b^{\prime}),\,\,\text{where}\,\, a=b^{\prime}.\label{cp-h-s-t1}
\end{equation}

With an analysis similar to the analysis in Section 3.1 of~\cite{ourpaper}, it can be proved that the following inequality applies to any correct choice of a convolutional encoder that implements either of the transformations in~(\ref{cn-cn-s-t1}-\ref{cp-h-s-t1}):
\begin{equation}
\sigma\leq \tau^{\prime}, \label{s-t}
\end{equation}
where $\sigma$ and $\tau^{\prime}$ denote the frame index of the source qubit of the first gate and the frame index of the target qubit of the second gate in a convolutional encoder respectively.
We call the inequality in~(\ref{s-t}), \emph{source-target constraint.}

As an example, the gate strings of the pearl-necklace encoder, $\overline{\text{CPHASE}}(2,3D)\overline{\text{CNOT}}(1,2D)$, (Figure \ref{constraintex}(a))
have source-target non-commutativity. A correct choice of convolutional encoder is (the encoder depicted over a first arrow in the Figure \ref{constraintex}):
\begin{equation}
\text{CPHASE}(2,3)(1,0)\text{CNOT}(1,2)(2,1).
\end{equation}
In this case $\sigma=1\leq \tau^{\prime}=1.$ Since the source-target constraint is satisfied the remaining gates after the chosen set in Figure\ref{constraintex}(b) can be shifted to the right. Repeated application of the procedure in (b) realizes a convolutional encoder representation from
a pearl-necklace encoder(Figure\ref{constraintex}(c)).
 \begin{figure}
[ptb]
\begin{center}
\includegraphics[
natheight=9.879900in,
natwidth=19.139999in,
height=2.35in,
width=5.03in
]
{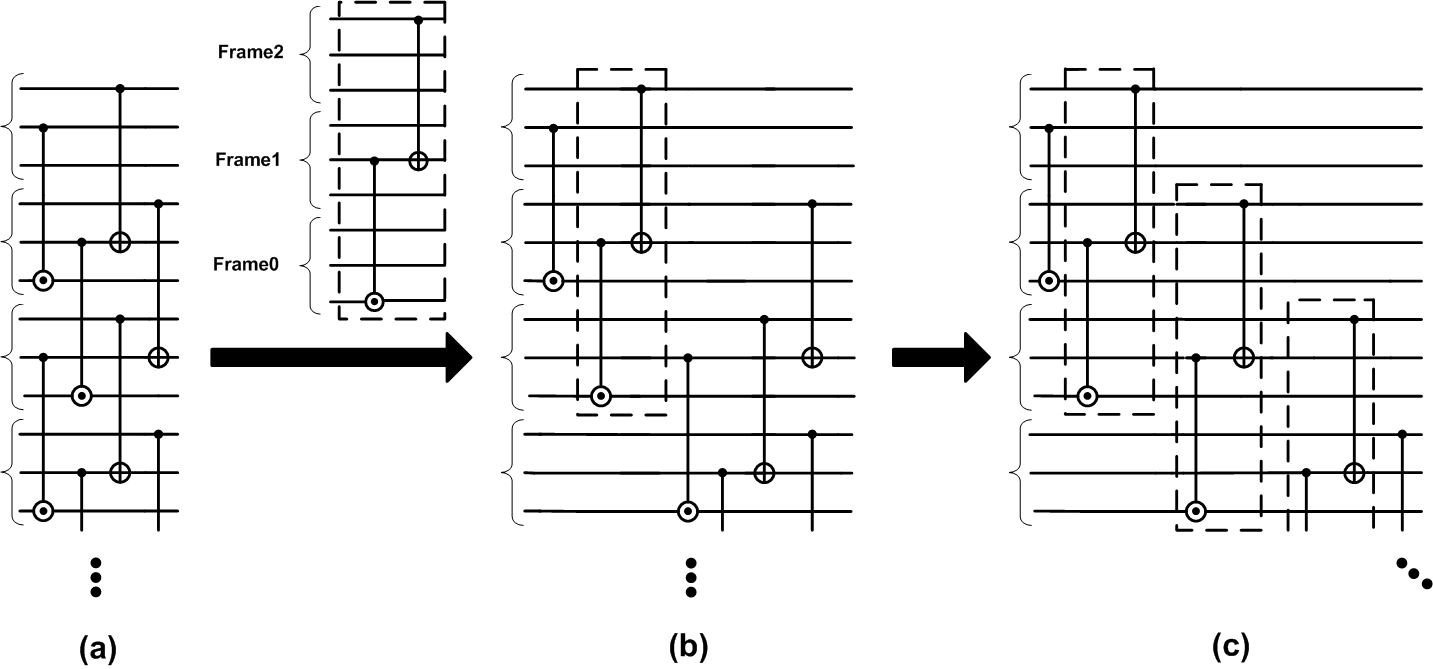}
\end{center}
\caption{Finding a correct choice for a two non-commutative gate strings . (a) The pearl-necklace encoder consists of
the gate strings
$\overline{\text{CPHASE}}(2,3D)\overline{\text{CNOT}}(1,2D)$, which have source-target non-commutativity.
(b) The
rearrangement of gates after the first three by shifting them
to the right. (c) The repeated application of the procedure in (b) realizes a
convolutional encoder from a pearl-necklace encoder.}
\label{constraintex}
\end{figure}

The following Boolean function is used to determine whether this type of non-commutativity exists for two gate strings:
\[\text{Source-Target}\left(\overline{U}_i,\overline{U}_j \right).\]
This function takes two gate strings $\overline{U}_i$ and $\overline{U}_j$ as
input. It returns TRUE if $\overline{U}_i$ and
$\overline{U}_j$ have source-target
non-commutativity and returns FALSE otherwise.
\subsubsection{Target-source non-commutativity}
It is obvious that the gate strings in~(\ref{cn-cn-t-s1}-\ref{h-cp-t-s1}) do not commute. In all of them, the index of each target qubit in the first gate string
is the same as the index of each source qubit in the second gate string. Therefore we call this type of non-commutativity, \emph{target-source non-commutativity}.
\begin{equation}
\overline{\text{CNOT}}(a,bD^{l})\overline{\text{CNOT}}(a^{\prime},b^{\prime}D^{l^{\prime}}),\,\,\text{where}\,\, b=a^{\prime}, \label{cn-cn-t-s1}
\end{equation}
\begin{equation}
\overline{\text{CNOT}}(a,bD^{l}) \overline{\text{CPHASE}}(a^{\prime},b^{\prime}D^{l^{\prime}}),\,\, \text{where}\,\, b=a^{\prime},\label{cn-cp-t-s1}
\end{equation}
\begin{equation}
\overline{H}(b)\overline{\text{CNOT}}(a^{\prime},b^{\prime}D^{l^{\prime}}),\,\,\text{where}\,\, b=a^{\prime}, \label{h-cn-t-s1}
\end{equation}
\begin{equation}
\overline{H}(b)\overline{\text{CPHASE}}(a^{\prime},b^{\prime}D^{l^{\prime}}),\,\,\text{where}\,\, b=a^{\prime}.\label{h-cp-t-s1}
\end{equation}

With an analysis similar to the analysis in Section 3.1 of~\cite{ourpaper}, it can be proved that the following inequality applies to any correct choice of a convolutional encoder that implements either of the transformations in~(\ref{cn-cn-t-s1}-\ref{h-cp-t-s1}):
\begin{equation}
\tau\leq \sigma^{\prime},\label{t-s}
\end{equation}
where $\tau$ and $\sigma^{\prime}$ denote the frame index of the target qubit of the first gate and the frame index of the source qubit of the second gate in a convolutional encoder respectively. We call the inequality in~(\ref{t-s}), \emph{target-source constraint.}

The following Boolean function is used to determine whether target-source non-commutativity exists for two gate strings:
\[\text{Target-Source}\left(\overline{U}_i,\overline{U}_j \right).\]
This function takes two gate strings $\overline{U}_i$ and $\overline{U}_j$ as
input. It returns TRUE if $\overline{U}_i$ and
$\overline{U}_j$ have target-source
non-commutativity and returns FALSE otherwise.
\subsubsection{Target-target non-commutativity}
It is obvious that the gate strings in~(\ref{cp-cn-t-t1}-\ref{h-p-t-t1}) do not commute. In all of them, the index of each target qubit in the first gate string is the same as the index of each target qubit in the second gate string. Therefore we call this type of non-commutativity, \emph{target-target non-commutativity}.
%\ref{cp-cn-t-t1}, \ref{cn-cp-t-t1}, \ref{cn-h-t-t1}, \ref{cp-h-t-t1}, \ref{h-cp-t-t1}, \ref{h-cn-t-t1}, \ref{cn-p-t-t1}, \ref{p-cn-t-t1}, \ref{p-h-t-t1} and \ref{h-p-t-t1}
\begin{equation}
\overline{\text{CPHASE}}(a,bD^{l})\ \overline{\text{CNOT}}(a^{\prime},b^{\prime}D^{l^{\prime}}),\,\,\text{where}\,\,b=b^{\prime}, \label{cp-cn-t-t1}
\end{equation}
\begin{equation}
\overline{\text{CNOT}}(a,bD^{l}) \overline{\text{CPHASE}}(a^{\prime},b^{\prime}D^{l^{\prime}}),\,\, \text{where}\,\, b=b^{\prime},\label{cn-cp-t-t1}
\end{equation}
\begin{equation}
\overline{\text{CNOT}}(a,bD^{l})\overline{H}(b),\,\,\text{where}\,\, b=b^{\prime}, \label{cn-h-t-t1}
\end{equation}
\begin{equation}
\overline{\text{CPHASE}}(a,bD^{l})\overline{H}(b^{\prime}),\,\,\text{where}\,\, b=b^{\prime}, \label{cp-h-t-t1}
\end{equation}
\begin{equation}
\overline{H}(b)\overline{\text{CNOT}}(a^{\prime},b^{\prime}D^{l^{\prime}}),\,\,\text{where}\,\, b=b^{\prime},\label{h-cn-t-t1}
\end{equation}
\begin{equation}
\overline{H}(b)\overline{\text{CPHASE}}(a^{\prime},b^{\prime}D^{l^{\prime}}),\,\,\text{where}\,\, b=b^{\prime},\label{h-cp-t-t1}
\end{equation}
 \begin{equation}
 \overline{\text{CNOT}}(a,bD^{l})\overline{P}(b^{\prime}),\,\,\text{where}\,\, b=b^{\prime}, \label{cn-p-t-t1}
 \end{equation}
 \begin{equation}
 \overline{P}(b)\overline{\text{CNOT}}(a^{\prime},b^{\prime}D^{l^{\prime}}),\,\,\text{where}\,\, b=b^{\prime},\label{p-cn-t-t1}
 \end{equation}
 \begin{equation}
 \overline{P}(b)\overline{H}(b^{\prime}),\,\, \text{where}\,\, b=b^{\prime},\label{p-h-t-t1}
 \end{equation}
\begin{equation}
\overline{H}(b)\overline{P}(b^{\prime}),\,\,\text{where} \,\, b=b^{\prime}. \label{h-p-t-t1}
\end{equation}

With a analysis similar to the analysis in Section 3.1 of~\cite{ourpaper}, it can be proved that the following inequality applies to any correct choice of a convolutional encoder that implements either of the transformations in~(\ref{cp-cn-t-t1}-\ref{h-p-t-t1}):
\begin{equation}
\tau\leq \tau^{\prime},\label{t-t}
\end{equation}
where $\tau$ and $\tau^{\prime}$ denote the frame index of the target qubit of the first gate and the frame index of the target qubit of the second gate in a convolutional encoder respectively. We call the inequality in (\ref{t-t}), \emph{target-target constraint.}
The following Boolean function is used to determine whether target-target non-commutativity exists for two gate strings:
\[\text{Target-Target}\left(\overline{U}_i,\overline{U}_j \right)=\text{TRUE}.\]
This function takes two gate strings $\overline{U}_i$ and $\overline{U}_j$ as
input. It returns TRUE if $\overline{U}_i$ and
$\overline{U}_j$ have target-target
non-commutativity and returns FALSE otherwise.

Consider the $j^{\text{th}}$ gate string, $\overline{U}_j$ in the encoder. It is important
to consider the gate strings preceding this one that do not commute with this gate string and categorize them based on the type of non-commutativity. Therefore we define the following sets:
\[
(\mathcal{S-T})_{j}  =\{i\mid
\text{Source-Target}(\overline{U}_i,\overline{U}_j)=\text{TRUE},i\in\{1,2,\cdots,j-1\}\},
\]
\[
(\mathcal{T-S})_{j}  =\{i\mid
\text{Target-Source}(\overline{U}_i,\overline{U}_j)=\text{TRUE},i\in\{1,2,\cdots,j-1\}\},
\]
\[
(\mathcal{T-T})_{j}  =\{i\mid
\text{Target-Target}(\overline{U}_i,\overline{U}_j)=\text{TRUE},i\in\{1,2,\cdots,j-1\}\}.
\]

\subsection{The proposed algorithm for finding minimal memory requirements for an arbitrary pearl-necklace encoder}
In this section we find the minimal-memory realization for an arbitrary pearl-necklace encoder which include all gate strings that are in shift-invariant Clifford group: Hadamard gates, Phase gates, controlled-phase and finite depth and infinite-depth controlled-NOT gate strings.
 To achieve this goal, we consider any non-commutativity that may exist for a particular gate and its preceding gates.
Suppose that
a pearl-necklace encoder features the following succession of $N$ gate strings:
\begin{equation}
\overline{U_1},\overline{U_2},\cdots,\overline{U_N}. \label{encoder}
\end{equation}
If the first gate string is $\overline{\text{CNOT}}(a_1,b_1D^{l_1}),l_1\geq 0$, the first gate in the convolutional encoder is
\begin{equation}
\text{CNOT}(a_1,b_1)(\sigma_1=l_1,\tau_1=0).\label{cn+}
\end{equation}
If the first gate string is $\overline{\text{CNOT}}(a_1,b_1D^{l_1}),l_1< 0$ the first gate in the convolutional encoder is
\begin{equation}
\text{CNOT}(a_1,b_1)(\sigma_1=0,\tau_1=|l_1|).\label{cn-}
\end{equation}
If the first gate string is $\overline{\text{CPHASE}}(a_1,b_1D^{l_1}),l_1\geq 0$
 the first gate in the convolutional encoder is
  \begin{equation}
  \text{CPHASE}(a_1,b_1)(\sigma_1=l_1,\tau_1=0).\label{cp+}
  \end{equation}
If the first gate string is $\overline{\text{CPHASE}}(a_1,b_1D^{l_1}),l_1< 0$ the first gate in the convolutional encoder is
 \begin{equation}
\text{CPHASE}(a_1,b_1)(\sigma_1=0,\tau_1=|l_1|).\label{cp-}
\end{equation}
If the first gate string is $\overline{H}(b_1)$ or $\overline{P}(b_1)$ the first gate in the convolutional encoder is as follows respectively:
\begin{equation}
H(b_1)(0),\label{h}
\end{equation}
\begin{equation}
P(b_1)(0).\label{p}
\end{equation}
For the
target indices of each gate $j$ where $2 \leq j \leq N$, we should choose a value for $\tau_{j}$ that satisfies all the
 constraints that the gates preceding it impose.

First consider $\overline{U_j}$ is the $\overline{\text{CNOT}}$ or $\overline{\text{CPHASE}}$ gate, then
the following inequalities must be satisfied to target index of $\overline{U_j},$ $\tau_{j}$:
%Two CNOT and CPHASE gates don\primet commute if and only if the index of the source qubit of the CPHASE gate is the same as the
%index of the target of the CNOT gate.

By applying the source-target constraint in (\ref{s-t}) we have:
\begin{align}
\sigma_{i}  &  \leq\tau_{j}\,\,\,\,\,\,\forall i\in(\mathcal{S-T})_{j}, \nonumber\\
\therefore\,\,\,\tau_{i}+l_i  &  \leq\tau_{j} \,\,\,\,\,\,\forall i\in(\mathcal{S-T})_{j},\nonumber\\
  \therefore\,\,\max\{\tau_{i}+l_{i}\}_{i\in(\mathcal{S-T})_{j}} &  \leq\tau_{j},        \label{cn-cp-s-t}
\end{align}
by applying the target-source constraint in (\ref{t-s}) we have:
\begin{align}
\tau_{i}  &  \leq\sigma_{j}\,\,\,\,\,\,\forall i\in(\mathcal{T-S})_{j}\nonumber\\
\therefore\,\,\,\tau_{i}  &  \leq\tau_{j}+l_j \,\,\,\,\,\,\forall i\in(\mathcal{T-S})_{j}, \nonumber\\
\therefore\,\,\,\tau_{i}-l_j  &  \leq\tau_{j}\,\,\,\,\,\,\forall i\in(\mathcal{T-S})_{j},\nonumber\\
\therefore\,\,\max\{\tau_{i}-l_{j}\}_{i\in(\mathcal{T-S})_{j}} &  \leq\tau_{j}.
 \label{cn-cp-t-s}
\end{align}
By applying the target-target constraint in (\ref{t-t}) we have:
\begin{align}
\tau_{i}  &  \leq\tau_{j}\,\,\,\,\,\,\forall i\in(\mathcal{T-T})_{j}\nonumber\\
\therefore\,\,\max\{\tau_{i}\}_{i\in(\mathcal{T-T})_{j}} &  \leq\tau_{j}. \label{cn-cp-t-t}
\end{align}
The following constraint applies to the frame index $\tau_j$ of the target qubit by applying (\ref{cn-cp-s-t}-\ref{cn-cp-t-t}):
%$\ref{cn-cp-s-t}$, $\ref{cn-cp-t-s}$ and $\ref{cn-cp-t-t}$:
\begin{align}
\max\{{\{\tau_i+l_i\}}_{i\in(\mathcal{S-T})_{j}},{\{\tau_i-l_j\}}_{i\in (\mathcal{T-S})_{j}},\{\tau_{i}\}_{i\in (\mathcal{T-T})_{j}}\} \leq\tau_{j}.\,\,\,\,\,\,
\end{align}
Thus, the minimal value for $\tau_{j}$ (which corresponds to the minimal-memory realization) that satisfies all the constraints is:
\begin{align}
\tau_{j}=\max\{{\{\tau_i+l_i\}}_{i\in(\mathcal{S-T})_{j}},{\{\tau_i-l_j\}}_{i\in (\mathcal{T-S})_{j}},\{\tau_{i}\}_{i\in (\mathcal{T-T})_{j}}\} .\,\,\,\,\,\, \label{cn-cp-min1}
\end{align}
It can be easily shown that there is no constraint for the frame index $\tau_{j}$ if the gate
string $\overline{U_j}$ commutes with all
previous gate strings. Thus if $l_j\geq0$ we choose the frame index $\tau
_{j}$ as follows:%
\begin{equation}
\tau_{j}=0.\label{cn-cp+min2}%
\end{equation}
and if $l_j<0$ we choose $\tau_{j}$ as follows:
\begin{equation}
\tau_{j}=|l_j|.\label{cn-cp-min2}%
\end{equation}
If $l_j\geq0,$ a good choice for the frame index $\tau_{j},$ by considering (\ref{cn-cp-min1}) and (\ref{cn-cp+min2}) is as follows:
\begin{align}
\tau_{j}=\max\{0,{\{\tau_i+l_i\}}_{i\in(\mathcal{S-T})_{j}},{\{\tau_i-l_j\}}_{i\in (\mathcal{T-S})_{j}},\{\tau_{i}\}_{i\in (\mathcal{T-T})_{j}}\} .\,\,\,\,\,\, \label{cn-cp+}
\end{align}
and if $l_j<0,$ a good choice for the frame index $\tau_{j},$ by considering (\ref{cn-cp-min1}) and (\ref{cn-cp-min2}) is as follows:
\begin{align}
\tau_{j}=\max\{|l_j|,{\{\tau_i+l_i\}}_{i\in(\mathcal{S-T})_{j}},{\{\tau_i-l_j\}}_{i\in (\mathcal{T-S})_{j}},\{\tau_{i}\}_{i\in (\mathcal{T-T})_{j}}\} .\,\,\,\,\,\, \label{cn-cp-}
\end{align}
Now consider $\overline{U_j}$ is the $\overline{H}$, then
the following inequalities must be satisfied to target index of $\overline{U_j},$ $\tau_{j}$:

By applying the source-target constraint in (\ref{s-t}) we have:
\begin{align}
\sigma_{i}  &  \leq\tau_{j}\,\,\,\,\,\,\forall i\in(\mathcal{S-T})_{j}, \nonumber\\
\therefore\,\,\,\tau_{i}+l_i  &  \leq\tau_{j} \,\,\,\,\,\,\forall i\in(\mathcal{S-T})_{j},\nonumber\\
  \therefore\,\,\max\{\tau_{i}+l_{i}\}_{i\in(\mathcal{S-T})_{j}} &  \leq\tau_{j}.        \label{h-p-s-t}
\end{align}
By applying target-target constraint in (\ref{t-t}) we have:
\begin{align}
\tau_{i}  &  \leq\tau_{j}\,\,\,\,\,\,\forall i\in(\mathcal{S-T})_{j}\nonumber\\
\therefore\,\,\max\{\tau_{i}\}_{i\in(\mathcal{T-T})_{j}} &  \leq\tau_{j}, \label{h-p-t-t}
\end{align}
The following constraint applies to the frame index $\tau_j$ of the target qubit  by applying ($\ref{h-p-s-t}$) and ($\ref{h-p-t-t}$):
\begin{align}
\max\{{\{\tau_i+l_i\}}_{i\in(\mathcal{S-T})_{j}},\{\tau_{i}\}_{i\in (\mathcal{T-T})_{j}}\} \leq\tau_{j}\nonumber.
\end{align}
Thus, the minimal value for $\tau_{j}$ (which corresponds to the minimal-memory realization) that satisfies all the constraints is:
\begin{align}
\tau_{j}=\max\{{\{\tau_i+l_i\}}_{i\in(\mathcal{S-T})_{j}},\{\tau_{i}\}_{i\in (\mathcal{T-T})_{j}}\} .\label{H-min1}
\end{align}
It can be easily shown that there is no constraint for the frame index $\tau_{j}$ if the gate
string $\overline{U_j}$ commutes with all
previous gate strings. Thus, in this case, we choose the frame index $\tau
_{j}$ as follows:%
\begin{equation}
\tau_{j}=0.\label{H-min2}%
\end{equation}
A good choice for the frame index $\tau_{j},$ by considering (\ref{H-min1}) and (\ref{H-min2}) is as follows:
\begin{align}
\tau_{j}=\max\{{0,\{\tau_i+l_i\}}_{i\in(\mathcal{S-T})_{j}},\{\tau_{i}\}_{i\in (\mathcal{T-T})_{j}}\}.
\end{align}
Now consider $\overline{U_j}$ is the $\overline{P}$, then
by applying the target-target constraint in (\ref{t-t}), the following inequality must be satisfied to target index of $\overline{U_j},$ $\tau_{j}$:
\begin{align}
\tau_{i}  &  \leq\tau_{j}\,\,\,\,\,\,\forall i\in(\mathcal{T-T})_{j}\nonumber\\
\therefore\,\,\max\{\tau_{i}\}_{i\in(\mathcal{T-T})_{j}} &  \leq\tau_{j}\nonumber.
\end{align}
Thus, the minimal value for $\tau_{j}$ (which corresponds to the minimal-memory realization) that satisfies all the constraints is:
\begin{align}
\tau_{j}=\max\{\tau_{i}\}_{i\in(\mathcal{T-T})_{j}}. \label{P-min1}
\end{align}
It can be easily shown that there is no constraint for the frame index $\tau_{j}$ if the gate
string $\overline{U_j}$ commutes with all
previous gate strings. Thus, in this case, we choose the frame index $\tau_{j}$ as follows:%
\begin{equation}
\tau_{j}=0.\label{P-min2}%
\end{equation}
A good choice for the frame index $\tau_{j},$ by considering (\ref{P-min1}) and (\ref{P-min2}) is as follows:
\begin{align}
\tau_{j}=\max\{0,\{\tau_{i}\}_{i\in (\mathcal{T-T})_{j}}\}.
\end{align}
\subsubsection{Construction of the non-commutativity graph}
We introduce the notion of a \emph{non-commutative} graph, $\mathcal{G}$ in order to find the
best values for the target qubit
frame indices. The graph is a weighted, directed acyclic graph constructed
from the non-commutativity relations of the gate strings
in~(\ref{encoder}). Algorithm \ref{graph} presents pseudo code for constructing the non-commutativity graph.
\begin{algorithm}
\caption{Algorithm for determining the non-commutativity graph $\mathcal{G}$ for general case}
\label{graph}
\begin{algorithmic}
\STATE$N \gets$ Number of gate strings in the pearl-necklace encoder.
\STATE Draw a {\bf START} vertex.
\FOR{$j := 1$ to $N$}
\STATE Draw a vertex labeled $j$ for the $j^{\text{th}}$ gate string $\overline{U}_j$
\IF{$j\in (I_{\text{CNOT}}^{-}\cup I_{\text{CPHASE}}^{-})$}
\STATE DrawEdge({\bf START}, $j$, $|l_j|$)
\ELSE
 \STATE DrawEdge({\bf START}, $j$, 0)
 \ENDIF
\FOR{$i$ := $1$ to $j-1$}
\IF{$j \in (I_{\text{CNOT}}^{+}\cup I_{\text{CPHASE}}^{+}\cup I_{\text{CNOT}}^{-}\cup I_{\text{CPHASE}}^{-}) $}
\IF{$i\in (\mathcal{S-T})_{j}$}
\STATE DrawEdge($i,j,l_i$ )
\ENDIF
\IF{$i\in (\mathcal{T-S})_{j}$}
\STATE DrawEdge($i,j,-l_j$)
\ENDIF
\IF{$i\in (\mathcal{T-T})_{j}$}
\STATE DrawEdge($i,j,0$)
\ENDIF
\ELSE
\IF{$j\in I_{H}$}
\IF{$i\in (\mathcal{S-T})_{j}$}
\STATE DrawEdge($i,j,l_i$)
\ENDIF
\IF{$i\in (\mathcal{T-T})_{j}$}
\STATE DrawEdge($i,j,0$)
\ENDIF
%\ENDIF
\ELSE
\IF{$i\in (\mathcal{T-T})_{j}$}
\STATE DrawEdge($i,j,0$)
\ENDIF
\ENDIF
\ENDIF
\ENDFOR
\ENDFOR
\STATE Draw an {\bf END} vertex.
\FOR{$j$ := $1$ to $N$}
\IF{$j\in (I_{\text{CNOT}}^+ \cup I_{\text{CPHASE}}^+$)} \STATE DrawEdge($j$,{\bf END}, $l_j$)
\ELSE \STATE DrawEdge($j$,{\bf END}, $0$)
\ENDIF
\ENDFOR
\end{algorithmic}
\end{algorithm}
 $\mathcal{G}$ consists of $N$ vertices, labeled
$1,2,\cdots,N$, where the $j^{\text{th}}$ vertex corresponds to the
$j^{\text{th}}$ gate string $\overline{U}_j$.
It also has two dummy vertices, named \textquotedblleft
START\textquotedblright\ and \textquotedblleft END.\textquotedblright%
\ DrawEdge$\left(  i,j,w\right)  $ is a function that draws a directed edge
from vertex $i$ to vertex $j$ with an edge weight equal to $w$.
\subsubsection{The longest path gives the minimal memory requirements}
Theorem~\ref{thm:longest-path-is-memory} below states that the weight of the
longest path from the START vertex to the END vertex is equal to the minimal memory
required for a convolutional encoder implementation.
\begin{theorem}
\label{thm:longest-path-is-memory}The weight $w$\ of the longest path from
the START vertex to END vertex in the non-commutativity graph $\mathcal{G}$ is equal
to the minimal memory $L$ that the convolutional encoder requires.
\end{theorem}
\begin{proof}
We first prove by induction that the weight $w_{j}$ of the longest path from
the START vertex to vertex $j$ in the non-commutativity graph $\mathcal{G}$ is%
\begin{equation}
w_{j}=\mathbb{\tau}_{j}. \label{eq:Gwlp}%
\end{equation}
Based on the algorithm, a zero-weight edge connects the START vertex to the first vertex, if $1\in(I_{\text{CNOT}}^{+}\cup I_{\text{CPHASE}}^{+}\cup I_{H}\cup I_{P})$ and in this case based on (\ref{cn+}), (\ref{cp+}), (\ref{h}) and (\ref{p}), $\tau_{1}=0$ therefore
$w_{1}=\tau_{1}=0$. An edge with the weight equal to $|l_1|$ connects the START vertex to the first gate if $1 \in (I_{\text{CNOT}}^{-}\cup I_{\text{CPHASE}}^{-}),$ and based on (\ref{cn-}) and (\ref{cp-}),  $\tau_{1}=|l_1|$ therefore
$w_{1}=\tau_{1}=|l_1|$.
Thus the base step holds for the target index of the first
 gate in a minimal-memory convolutional encoder. Now suppose the property
holds for the target indices of the first $k$ gates in the convolutional
encoder:%
\[
w_{j}=\mathbb{\tau}_{j}\,\,\,\,\,\,\forall j : 1\leq j\leq k.
\]
Suppose we add a new gate string $\overline{U}_{k+1}$ to the pearl-necklace encoder, and
Algorithm~\ref{graph} then adds a new vertex $k+1$ to the
graph $\mathcal{G}.$ Suppose $(k+1)\in (I^{+}_\text{CNOT}\cup I^{+}_\text{CPHASE})$. The following edges are added to $\mathcal{G}$:
\begin{enumerate}
\item A zero-weight edge from the START vertex to vertex $k+1$.

\item An $l_{i}$-weight edge from each vertex $\{i\}_{i\in\mathcal({S-T})_{k+1}}$
to vertex $k+1$.
\item A $-l_{k+1}$-weight edge from each vertex $\{i\}_{i\in\mathcal({T-S})_{k+1}%
}$ to vertex $k+1$.
\item A zero-weight edge from each vertex $\{i\}_{i\in(\mathcal{T-T})_{k+1}%
}$  to vertex $k+1$.
\item An $l_{k+1}$-weight edge from vertex $k+1$ to the END vertex.
\end{enumerate}
It is clear that the following relations hold:%
\begin{align}
w_{k+1}  &  =\max\{0,\{w_{i}+l_{i}\}_{i\in(\mathcal{S-T})_{k+1}},\{w_{i}%
-l_{k+1}\}_{i\in(\mathcal{T-S})_{k+1}},\{w_{i}\}_{i\in(\mathcal{T-T})_{k+1}}\},\nonumber\\
&  =\max\{0,\{\mathbb{\tau}_{i}+l_{i}\}_{i\in(\mathcal{S-T})_{k+1}},\{\mathbb{\tau
}_{i}-l_{k+1}\}_{i\in(\mathcal{T-S})_{k+1}},\{\mathbb{\tau}_{i}\}_{i\in(\mathcal{T-T})_{k+1}}\}. \label{Gwlp2}%
\end{align}
 By applying
(\ref{cn-cp+}) and (\ref{Gwlp2}) we have:%
\[
w_{k+1}=\tau_{k+1}.
\]
In a similar way we can show that if the $U_{k+1}$ is any other gate string of Clifford shift-invariant:
\[w_{k+1}=\tau_{k+1}.\]
The proof of the theorem then follows by considering the following equalities:%
\begin{align}
w  &  =\max\{\max_{i\in (I_{\text{CNOT}}^{+}\cup I_{\text{CPHASE}}^{+}\cup I_{H}\cup I_{P})}\{w_{i}+l_{i}\}, \max_{i\in (I_{\text{CNOT}}^{-}\cup I_{\text{CPHASE}}^{-})}\{w_i\}\} \nonumber\\
&  =\max\{\max_{i\in (I_{\text{CNOT}}^{+}\cup I_{\text{CPHASE}}^{+}\cup I_{H}\cup I_{P})}\{\tau_{i}+l_{i}\}, \max_{i\in (I_{\text{CNOT}}^{-}\cup I_{\text{CPHASE}}^{-})}\{\tau_{i}\}\}\nonumber\\
&  =\max\{\max_{i\in (I_{\text{CNOT}}^{+}\cup I_{\text{CPHASE}}^{+}\cup I_{H}\cup I_{P})}\{\sigma_{i}\}, \max_{i\in (I_{\text{CNOT}}^{-}\cup I_{\text{CPHASE}}^{-})}\{\tau_{i}\}\nonumber\}.
\end{align}
The first equality holds because the longest path in the graph is the maximum of the
weight of the path to the $i^{\text{th}}$ vertex summed with the weight of the
edge to the END\ vertex. The second equality follows by applying
(\ref{eq:Gwlp}). The final equality follows because ${{\sigma}_{i}}=\tau
_{i}+l_{i}$. It is clear that\[\max\{\max_{i\in (I_{\text{CNOT}}^{+}\cup I_{\text{CPHASE}}^{+}\cup I_{H}\cup I_{P})}\{\sigma_{i}\}, \max_{i\in (I_{\text{CNOT}}^{-}\cup I_{\text{CPHASE}}^{-})}\{\tau_{i}\}\},\] is equal to minimal required memory for a minimal-memory convolutional
encoder, hence the theorem holds.
\end{proof}
The final task is to determine the longest path in $\mathcal{G}$. Finding the
longest path in a graph, in general is an NP-complete problem, but in a
weighted, directed acyclic graph requires linear time with dynamic
programming~\cite{cormen}. The non-commutativity graph $\mathcal{G}$ is an acyclic graph because a directed edge
connects each vertex only to vertices for which its corresponding gate comes later in the pearl-necklace encoder.

The running time for the construction of the graph is quadratic in the number of gate strings in the pearl-necklace
encoder. Since in
Algorithm~\ref{graph}, the instructions in the
inner \textbf{for} loop requires constant time $O(1)$. The sum of iterations
of the \textbf{if} instruction in the $j^{\text{th}}$ iteration of the outer
\textbf{for} loop is equal to $j-1$. Thus the running time $T(N)$\ of
Algorithm~\ref{graph} is
\[
T(N)=\sum_{i=1}^{N}{\sum_{k=1}^{j-1}O(1)}=O(N^{2}).
\]
 \begin{figure}
[ptb]
\begin{center}
\includegraphics[
natheight=9.879900in,
natwidth=19.139999in,
height=3.7in,
width=6.37in
]
{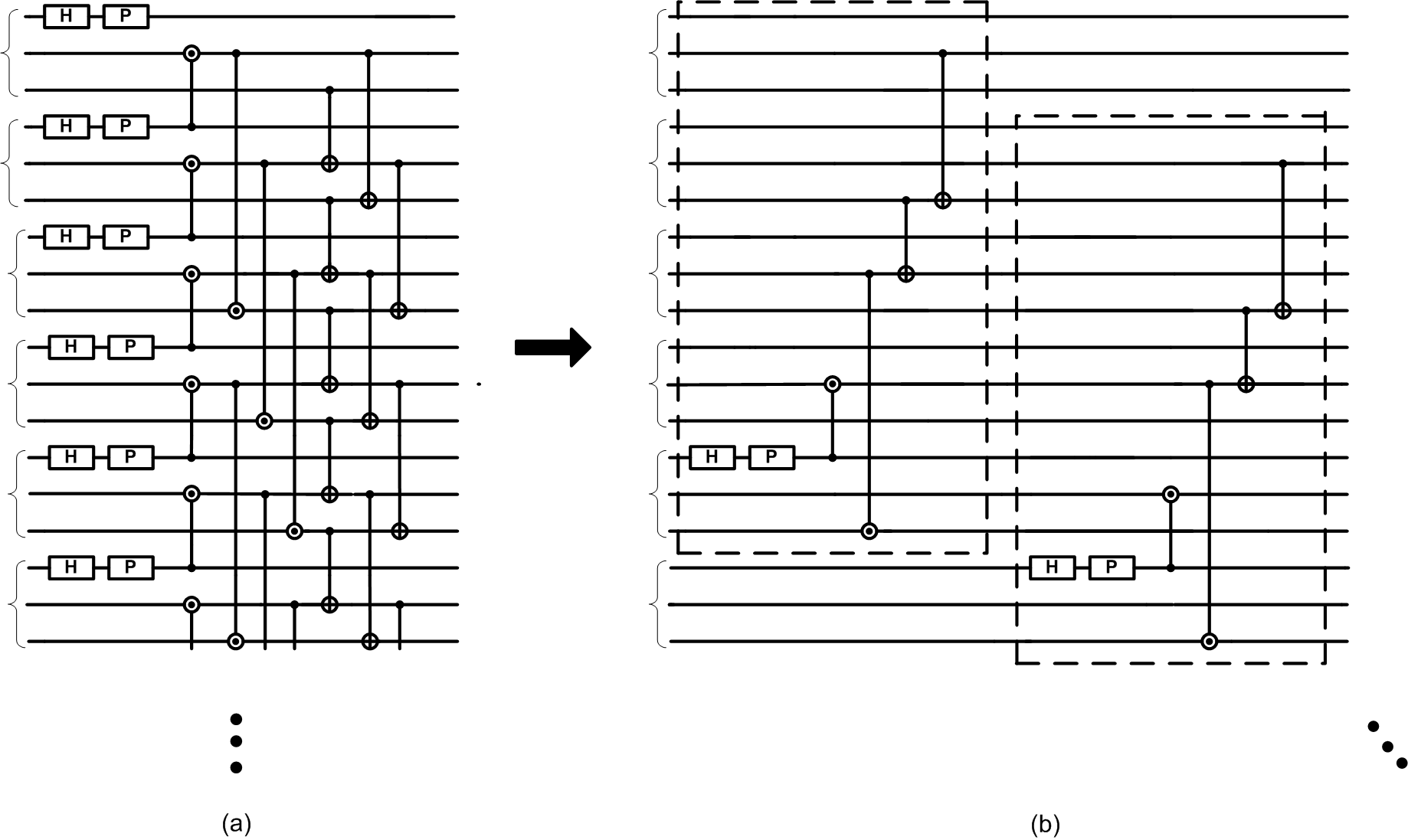}
\end{center}
\caption
{(a) pearl-necklace representation, and (b) minimal-memory convolutional encoder representation for example
1.}
\label{pearl-convo-ex}
\end{figure}%
\begin{figure}
[ptb]
\begin{center}
\includegraphics[
natheight=8.879900in,
natwidth=19.139999in,
height=3.213in,
width=4.61839in
]
{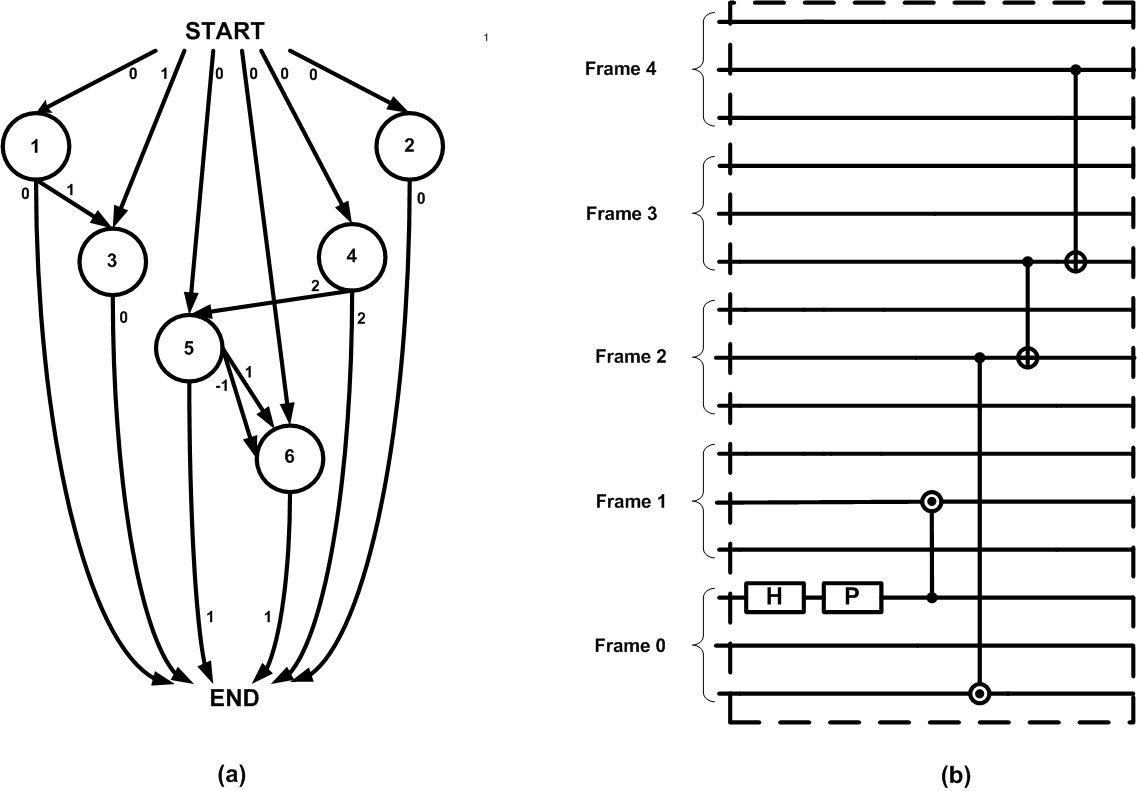}
\end{center}
\caption
{(a) The non-commutativity graph $\mathcal{G}$ and (b) a minimal-memory convolutional encoder for Example
1}
\label{graph-convo1}
\end{figure}%
 Example 1: Consider the following succession of gate strings in a pearl-necklace encoder(Figure~\ref{pearl-convo-ex}(a)):
\begin{align}
 &\overline{H}\left(  1\right)   P\left(  1\right)  \overline{\text{CPHASE}}%
\left(  1,2D^{-1}\right) \overline{\text{CPHASE}}%
\left(  2,3D^2\right)
  \overline{  \text{CNOT}}\left(
3,2D\right)  \overline{\text{CNOT}}\left(  2,3D\right)
\nonumber ,
\end{align}
Figure \ref{graph-convo1}(a) draws $\mathcal{G}$ for this pearl-necklace encoder,
after running Algorithm. The longest path through the graph is
\[
\text{START}\rightarrow4\rightarrow5\rightarrow6\rightarrow\text{END},
\]
with weight equal to four (0+2+1+1). Therefore the minimal memory for the convolutional encoder
is equal to four frames of
memory qubits. Also from inspecting the graph $\mathcal{G}$,
we can determine the locations for all the target qubit frame indices: $\tau_{1}=0,$ $\tau_{2}=0,$ $\tau_{3}=1,$ $\tau_{4}=0,$ $\tau_{5}=2,$ and $\tau_{6}=3.$
 Figure~\ref{graph-convo1}(b) depicts a
minimal-memory convolutional encoder for this example. Figure~\ref{pearl-convo-ex}(b) depicts minimal-memory convolutional encoder representation for the pearl-necklace encoder in Figure~\ref{pearl-convo-ex}(a).

\section{Conclusion}

%Some interesting open questions remain.
%The contribution
%here should be of use in the practical realization of encoders for quantum convolutional codes.
\label{conclu}
In this paper, we have proposed an algorithm to find a practical realization with a minimal 
memory requirement for a a pearl-necklace encoder of a general quantum convolutional code, 
which includes any gate string in the
shift-invariant Clifford group.
We have shown that the non-commutativity relations of gate strings in the encoder
determine the realization. We introduce a non-commutativity graph, whose each vertex 
corresponds to a gate string in the pearl-necklace encoder. The weighted edges represent 
non-commutativity relations in the encoder. Using the graph, the minimal-memory realization 
can be obtained. The weight of the longest path in the graph is equal to the minimal 
required memory of the encoder. The running time for the construction of the graph and 
finding the longest path is quadratic in the number of gate strings
in the pearl-necklace encoder.

As we mentioned in our previous paper~\cite{ourpaper}, an open question still remains.
The proposed algorithm begins with
a particular pearl-necklace encoder, and finds the minimal required memory for it. 
But one can start with polynomial description of convolutional code and find the minimal 
required memory for the code. There are two problems here to work on: 
(1) finding the pearl-necklace encoder with minimal-memory requirements among all 
pearl-necklace encoders that implement the same code, (2) constructing a repeated
unitary directly from the polynomial description of the code itself, and 
attempting to minimize the memory requirements of realizing this code.

\section*{Acknowledgements}
The authors would like to thank Mark M. Wilde for his helpful discussions, 
comments and careful reading of an earlier version of this paper.
\bibliographystyle{unsrt}
\bibliography{Ref}

\end{document}